\documentclass[a4paper,twoside,reqno]{amsart}
\usepackage{amsmath}
\usepackage{amsfonts}
\usepackage{amssymb}
\usepackage{amsthm}
\usepackage{newlfont}
\usepackage{graphicx}
\usepackage{amscd}
\usepackage{young}

\theoremstyle{plain}
\newtheorem{Th}{Theorem}[section]
\newtheorem{Cor}[Th]{Corollary}
\newtheorem{Lem}[Th]{Lemma}
\newtheorem{Prop}[Th]{Proposition}
\theoremstyle{definition}
\newtheorem{Def}{Definition}[section]
\newtheorem{Ex}{Example}[section]

\theoremstyle{remark}
\newtheorem*{Rem}{Remark}
\numberwithin{equation}{section}

\newcommand{\NN}{{\mathbb N}}

\newcommand{\DD}{{\mathbb D}}

\newcommand{\ZZ}{{\mathbb Z}}
\newcommand{\VV}{{\mathbb V}}
\newcommand{\RR}{{\mathbb R}}

\newcommand{\be}{\boldsymbol{e}}
\newcommand{\bs}{\boldsymbol{s}}

\newcommand{\bPhi}{\boldsymbol{\Phi}}

\begin{document}
	
	\title[Non-commutative multiple bi-orthogonal polynomials and integrability]
	{Non-commutative multiple bi-orthogonal polynomials:\\ formal approach and integrability}
	
	\author{Adam Doliwa}
	
	\address{Faculty of Mathematics and Computer Science\\
		University of Warmia and Mazury in Olsztyn\\
		ul.~S{\l}oneczna~54\\ 10-710~Olsztyn\\ Poland} 
	\email{doliwa@matman.uwm.edu.pl}
	\urladdr{http://wmii.uwm.edu.pl/~doliwa/}
	
	%
	
	\begin{abstract}
We define the non-commutative multiple bi-orthogonal polynomial systems, which simultaneously generalize the concepts of multiple orthogonality, matrix orthogonal polynomials and of the bi-orthogonality. 
 We present quasideterminantal expressions for such polynomial systems in terms of formal bi-moments. The normalization functions for such monic polynomials satisfy the non-commutative Hirota equations, while the polynomials provide solution of the corresponding linear system. This shows, in particular, that our polynomial systems form a part of the theory of integrable systems.  
 We study also a specialization of the problem to  non-commutative multiple orthogonal polynomials, what results in the corresponding Hankel-type quasideterminantal expressions in terms of the moments. Moreover, such a reduction allows to introduce in a standard way the discrete-time variable and gives rise to an integrable system which is non-commutative version of the multidimensional discrete-time Toda equations.

	\end{abstract}
	\maketitle
	
	\section{Introduction}

The standard theory of orthogonal polynomials \cite{Chihara,Ismail,Szego} provides important technical tool used to study various problems in theoretical physics or applied mathematics.  They can be encountered when solving, by separation of variables, partial differential equations of the classical field theory or the quantum mechanics~\cite{Nikiforov-Suslov-Uvarov}. By relation to continued fractions the orthogonal polynomials appear in number theory and combinatorics~\cite{Flajolet,Viennot}.
Another source are problems of probability theory leading~\cite{KarlinMcGregor,Schoutens,Iglesia} to the difference equations of hypergeometric type or various $q$--analogs of the above ones \cite{Ismail}. 
In mathematical physics they can be found in the spectral theory of operators in Hilbert space~\cite{Akhiezer,Szego,Geronimus}. More recently they were associated with the theory of representations of groups and algebras (including quantum groups) \cite{Vilenkin-Klimyk,Klimyk-Schmudgen}. Orthogonal polynomials on the unit circle are used to study quantum walks~\cite{CMGV,Doliwa-Siemaszko-QW}. They have found application in the theory of random matrices~\cite{FIK,Deift} and are used to construct special solutions to Painlev\'{e} equations~\cite{AdlervanMoerbeke,VanAssche}.

The idea of considering matrix-valued measures is due to Krein~\cite{Krein}, who was motivated by a moment problem arising in operator theory; see also more recent works~\cite{AptekarevNikishin,Geronimo}. Gradually, matrix orthogonal polynomials became useful tool to study problems arising from group representation theory~\cite{Grunbaum,Koelink}, quasi-birth-and-death processes~\cite{GrunbaumIglesia}, combinatorics~\cite{DuitsKuijlaars} or integrable systems~\cite{Miranian,Alvarez-FernandezPrietoManas,BranquinhoMorenoFradiManas,Shi-HaoLi,LSYZ}; see also \cite{SinapAssche,Damanik} for reviews. Formal theory of non-commutative orthogonal polynomials, which we will follow in our work, was given in \cite{NCSF}. It was formulated in terms of quasideterminants~\cite{Quasideterminants-GR1}, which are an important tool in non-commutative linear algebra.

Multiple orthogonal polynomials \cite{Aptekarev,NikishinSorokin,MF-VA} are a generalization of orthogonal polynomials in which the orthogonality is distributed among a number of weights. 
The development of the theory of multiple orthogonal polynomials in XXth century is summarized in~\cite{NikishinSorokin,Aptekarev}. Properties of multiple orthogonal polynomials for classical weights were described in~\cite{VanAsscheCoussement,AptekarevBranquinhoVanAssche}. Application to random matrices with external source and non-intersecting path ensembles was the subject of \cite{BleherKuijlaars,DaemsKuilaars,Kuijlaars}. Also some structural results generalizing those of orthogonal polynomials can be found in~\cite{CoussementVanAasche,Ismail,VanAsche-nn-mop,BranquinhoMorenoManas}.  In \cite{Sokal} the formal theory of multiple orthogonal polynomials with coefficients in a commutative ring is built with the starting point being the moments (of the measures) --- the approach well known~\cite{Chihara} in the standard theory of orthogonal polynomials. 

Recently, tho concept of orthogonality of a polynomial system was extended \cite{BertolaGekhtmanSzmigielski-JPA,BertolaGekhtmanSzmigielski-JAT} to bi-orthogonality, where the pairing between polynomials does not have to be symmetric. Such a generalization can be put in relation to previous works on skew orthogonal polynomials and their applications~\cite{AFNvM,AHvM}, see also more recent work \cite{ChangHeHuLi} on the subject. This idea has been recently applied to study matrix bi-orthogonal polynomials \cite{BranquinhoMorenoManas-MBOP,LSYZ} or even a version of multiple bi-orthogonal polynomials \cite{LiShenXiangYu} together with its symmetric and skew-symmetric reductions. 

The goal of the paper is twofold. Firstly, we unify the three above generalizations of the theory of orthogonal polynomials by proposing the concept of non-commutative multiple bi-orthogonal polynomials. Secondly, we present the theory of such polynomial systems as a part of the theory of integrable systems. In particular, we generalize to non-commutative case the well known connection between the orthogonal polynomials on the real line and the Toda lattice equations~\cite{Toda-TL,Hirota-2dT,Flaschka,Moser}, which form a paradigmatic example of an integrable system.
In recent papers~\cite{Alvarez-FernandezPrietoManas,AptekarevDerevyaginMikiVanAssche,FernandezManas,Doliwa-MOP} relation of the multiple orthogonal polynomials to integrable systems is studied. Also in the above-mentioned papers on bi-orthogonal polynomials and their matrix version the relation to integrability was essential.

Integrable discrete equations \cite{IDS} have emerged as a  field of study within mathematical physics and applied mathematics, providing insights into the behavior of discrete dynamical systems with remarkable properties. The concept of integrability goes beyond mere solvability and extends to the existence of a rich algebraic and geometric structure, abundance of conserved quantities and symmetries. This integrability property enables the development of powerful mathematical methods for studying behavior of the equations. The link between integrable discrete systems and orthogonal polynomials provides better understanding of their solutions, recursion relations, and symmetry properties.

In the present work we follow the approach to integrability initiated by Hirota, and developed by Sato and his school \cite{Sato,DKJM}. On the elementary level it is based on determinantal identities~\cite{Hirota-book} satisfied by the so called $\tau$-function; see \cite{Harnad} for modern presentation of its advanced aspects. Within this approach a particular role is played by Hirota's equation~\cite{Hirota}, known also as the discrete Kadomtsev-Petviashvili (KP) equation. It contains, as symmetry reductions or appropriate limits, many known integrable equations, see  \cite{DKJM,Miwa,Shiota,Zabrodin,KNS-rev,BialeckiDoliwa,Dol-Des,Dol-AN} for various aspects of the Hirota system both on the classical and quantum levels. 
In our research presented below we consider the non-commutative version of the Hirota system introduced by Nimmo~\cite{Nimmo-NCKP}.

 We give also, as a specification of our approach, the non-commutative version of recent results on multiple orthogonal polynomials presented in~\cite{AptekarevDerevyaginMikiVanAssche,Doliwa-MOP}, in particular that of the multidimensional discrete-time Toda lattice. Relation to integrability of the close connected subject~\cite{NikishinSorokin,Aptekarev,AptekarevKuijlaars,VanAsche-HPAO} of Hermite--Pad\'{e} (type I) approximation~\cite{BakerGraves-Morris} has been clarified recently in~\cite{Doliwa-Siemaszko-HP}, where determinantal identities were used to formulate corresponding reduction of the Hirota system on the level of $\tau$-function. 

The structure of the paper is as follows. In Section~\ref{sec:QD,NCH} to provide a motivation and background, we recall first standard facts on relation of the Toda lattice equations to orthogonal polynomials. We present also the non-commutative Hirota system, and then recall necessary information on quasideterminants. Then in Section~\ref{sec:NC-multiple-BOP} we define  non-commutative multiple bi-orthogonal polynomials and we provide their quasideterminantal representation. We discuss also the corresponding generalization of the three term recurrence formula, which can be interpreted as a linear problem for the non-commutative Hirota system, and leads to the corresponding solutions of the system. In Section~\ref{sec:NC-multiple-OP} we go down from bi-orthogonality to standard orthogonality condition obtaining the non-commutative multiple orthogonal polynomials, which also are new in the literature. Finally, in Section~\ref{sec:MTS}, we introduce discrete-time evolution on the level of their moments what results in the integrable non-commutative  multidimensional discrete-time Toda lattice equations. 

\section{Preliminaries}
\label{sec:QD,NCH}

\subsection{Orthogonal polynomials and the Toda lattice equations}

Let $\mu$ be a positive measure on the real line, consider the corresponding system of monic polynomials subject to orthogonality relations
\begin{equation*} \label{eq:orthogonality}
	\int_\mathbb{R} Q_r(x) Q_s(x) d\mu(x) = 0, \qquad r\neq s, \qquad Q_s(x) = x^s + \dots \; . 
\end{equation*}
The polynomials can be written down in terms of the moments
\begin{equation*} \label{eq:nu}
	\nu_k = \int_{\mathbb{R}} x^k d\mu(x), 
\end{equation*}
for which we assume that all they exist, as
\begin{equation}  \label{eq:Q-OP}
	Q_s(x) = \frac{1}{D_s} \left| \begin{matrix}
		\nu_0 & \nu_1 & \nu_2 & \dots & \nu_s \\
		\nu_1 & \nu_2 & \nu_3 & \dots & \nu_{s+1} \\
		\vdots & \vdots & \ddots  & & \vdots \\
		\nu_{s-1} & \nu_s & \nu_{s+1} & \dots & \nu_{2s-1} \\
		1 & x & x^2 & \dots & x^s
	\end{matrix} \right| , \quad \text{where} \quad 	D_s = \left| \begin{matrix}
		\nu_0 & \nu_1 &  \dots & \nu_{s-1} \\
		\nu_1 & \nu_2 &  \dots & \nu_{s} \\
		\vdots & \vdots & \ddots  & \vdots \\
		\nu_{s-1} & \nu_s &  \dots & \nu_{2s-2} 
	\end{matrix} \right| .
\end{equation}

It is well known that the orthogonality  conditions imply the three-term recurrence relation~\cite{Akhiezer,Szego,Ismail}
\begin{equation} \label{eq:3-term}
	xQ_s(x) = Q_{s+1}(x) + b_s Q_s(x) + a_s Q_{s-1}(x), 
\end{equation}
where the coefficients $a_s$, $b_s$ can be also written in terms of certain determinants built out of the moments.
Assume that measure undergoes evolution~\cite{Kac-vMoerbeke,Moser} of the form
\begin{equation*} \label{eq:mu-t}
	d\mu(x;t) = e^{-xt} d\mu(x), \qquad t \in \RR_+,
\end{equation*}
well known in the theory of continuous-time birth and death processes~\cite{LedermannReuter,KarlinMcGregor}. Then the moments evolve according to the law
\begin{equation*} 
	\dot{\nu}_s(t)= -\nu_{s+1}(t) ,
\end{equation*}
and the coefficients of the three-term relation satisfy the following system of differential equations
\begin{align*} 
	\dot{a}_s(t) & =  a_s(t)(b_{s-1}(t) - b_s(t)), \\
	\dot{b}_s(t) & =  a_s(t) - a_{s+1}(t),
\end{align*}
which are the Toda lattice equations~\cite{Toda-TL} in the form given by Flaschka~\cite{Flaschka}. Equivalently, the determinants $D_s(t)$ satisfy Toda lattice equations in the Hirota form 
\begin{equation*} \label{eq:CTTL}
	\ddot{D}_s(t) D_s(t) = \dot{D}_s(t)^2 + D_{s+1}(t) D_{s-1}(t).
\end{equation*}

When the measure undergoes the discrete-time evolution of the form
\begin{equation*}
	d\mu_t(x) = x^t d\mu(x), \qquad t\in \mathbb{N}_0,	
\end{equation*} 
then the moments evolve according to the law 
\begin{equation} \label{eq:mu-t}
	\nu_{s;t} = \nu_{s+t}, 
\end{equation}
and the corresponding determinants 
\begin{equation} \label{eq:D-t}
	D_{s;t} = 	 \left| \begin{matrix}
		\nu_t & \nu_{t+1} &  \dots & \nu_{t+s-1} \\
		\nu_{t+1} & \nu_{t+2} &  \dots & \nu_{t+s} \\
		\vdots & \vdots & \ddots  & \vdots \\
		\nu_{t+s-1} & \nu_{t+s} &  \dots & \nu_{t+2s-2} 
	\end{matrix} \right| ,
	\end{equation}
satisfy  the discrete-time Toda lattice equations \cite{Hirota-2dT}
\begin{equation} \label{eq:DTTL-D}
	D_{s-1;t+1} D_{s+1;t-1} = D_{s;t-1} D_{s;t+1} - D_{s;t}^2 \, .
\end{equation} 
Another standard form of the equations, which we will refer to in Section~\ref{sec:MTS}, is obtained when we replace, using the Jacobi--Desnanot identity~\cite{Hirota-book}, the three-term relation \eqref{eq:3-term} in constant time by its time-dependent versions~\cite{AptekarevDerevyaginMikiVanAssche} for the polynomials
\begin{equation}  \label{eq:Q-OP-t}
	Q_{s;t}(x) = \frac{1}{D_{s;t}} \left| \begin{matrix}
		\nu_t & \nu_{t+1} & \nu_{t+2} & \dots & \nu_{t+s} \\
		\nu_{t+1} & \nu_{t+2} & \nu_{t+3} & \dots & \nu_{t+s+1} \\
		\vdots & \vdots & \ddots  & & \vdots \\
		\nu_{t+s-1} & \nu_{t+s} & \nu_{t+ s+1} & \dots & \nu_{t+ 2s-1} \\
		1 & x & x^2 & \dots & x^s
	\end{matrix} \right| , 
\end{equation}
which take the form
\begin{align} \label{eq:lin-A}
	xQ_{s;t+1}(x) & = Q_{s+1;t}(x) + Q_{s,t}(x) A_{s;t}, \qquad A_{s;t} = \frac{ D_{s+1;t+1} D_{s;t} }{D_{s+1;t} D_{s;t+1}},\\ \label{eq:lin-B}
	Q_{s;t+1} & = Q_{s;t} - Q_{s-1;t+1} B_{s;t}, \qquad B_{s;t} = \frac{ D_{s+1;t} D_{s-1;t+1}} {D_{s;t} D_{s;t+1}}.
\end{align}
Then the compatibility condition between the above system gives equations
\begin{equation} \label{eq:dt-T-AB}
	A_{s;t+1}  + B_{s;t+1} = A_{s;t} + B_{s+1;t}, \qquad A_{s-1;t+1} B_{s;t+1} = A_{s;t} B_{s;t},
\end{equation}
which follow from the definitions given in~\eqref{eq:lin-A}-\eqref{eq:lin-B} and by equation \eqref{eq:DTTL-D}.

\subsection{The non-commutative Hirota system}
	The non-commutative Hirota system was introduced by Nimmo \cite{Nimmo-NCKP} as the compatibility conditions of the linear equations
	\begin{equation*} 
		\bPhi_{\bs + \be_i} - \bPhi_{\bs + \be_j} =  \bPhi_{\bs} U^{(ij)}_{\bs},  
		\qquad i \ne j ,
	\end{equation*}
	for the wave-function $\bPhi \colon \ZZ^N \to \VV$ taking values in a (right) vector space over a division ring $\DD$; here indices $i,j$ range from $1$ to $N$, $\bs \in \ZZ^N$ and $\be_i$ is an element of the canonical basis of the lattice. Then the functions $U^{(ij)}: \ZZ^N\to \DD$ are solutions of the equations
	\begin{equation*}
		U^{(ij)}_{\bs} +U^{(ji)}_{\bs} = 0, \quad
		U^{(ij)}_{\bs} +U^{(jk)}_{\bs} +U^{(ki)}_{\bs} =0, \quad
		U^{(ki)}_{\bs} U^{(kj)}_{\bs+\be_i} = U^{(kj)}_{\bs} U^{(ki)}_{\bs+\be_j}  ,
	\end{equation*}
with $ i,j,k$ distinct.
	By the third part of the above system there exist potentials $\rho^{(i)}:\ZZ^N\to \DD$, $i=1,\dots ,N$, such that
	\begin{equation*}
		U^{(ij)}_{\bs} = [\rho^{(i)}_{\bs}]^{-1} \rho^{(i)}_{\bs+\be_j}	,
	\end{equation*}
	and then the first two parts give the potential form 
	\begin{equation} 
		\label{eq:NC-H-rho}
		[\rho^{(i)}_{\bs}]^{-1}	\rho^{(i)}_{\bs+\be_j} + [\rho^{(j)}_{\bs}]^{-1}	\rho^{(j)}_{\bs+\be_i} = 0, \qquad 
		[\rho^{(i)}_{\bs}]^{-1}	\rho^{(i)}_{\bs+\be_j} + [\rho^{(j)}_{\bs}]^{-1}	\rho^{(j)}_{\bs+\be_k} + [\rho^{(k)}_{\bs}]^{-1}	\rho^{(k)}_{\bs+\be_i} = 0,
	\end{equation}
	with $i,j,k$ distinct, of the non-commutative Hirota system.
	
	When $\DD=\Bbbk$ is \emph{commutative} then the functions $\rho^{(i)}$ can be parametrized in terms of a single potential $\tau : \ZZ^N \to \Bbbk$ such that
	\begin{equation*} 
		\rho^{(i)}_{\bs} = (-1)^{\sum_{k<i} s_k}
		\frac{\tau_{\bs + \be_i}}{\tau_{\bs}}
	\end{equation*}
	and the nonlinear system reads 
	\begin{equation} \label{eq:Hirota-com}
		\tau_{\bs + \be_i} \tau_{\bs + \be_j + \be_k} - 
		\tau_{\bs + \be_j} \tau_{\bs + \be_i + \be_k} + 
		\tau_{\bs + \be_k} \tau_{\bs + \be_i + \be_j}=0,
		\qquad i< j <k ,
	\end{equation}
	which is the standard form of the Hirota system~\cite{Hirota,Miwa}. 

\begin{Rem}
As it was mentioned above, the Hirota system contains most known discrete integrable equations as specifications or symmetry reductions~\cite{KNS-rev}. The role of this equation for the theory of integrable systems can be compared to the importance (by Cayley's theorem) of the symmetric group in the theory of groups. Interestingly, the Hirota system is deeply connected with root lattices of type $A$ \cite{Dol-AN}, whose affine Weyl groups are semidirect products of translations and permutations \cite{Humphreys}. This has strong relations with $\mathfrak{gl}(\infty)$ structure of symmetries \cite{DKJM} of the KP (or A-KP) hierarchy of integrable equations. The integrable hierarchies of $\mathfrak{so}(\infty)$ or $\mathfrak{sp}(\infty)$ structure of symmetries lead to the discrete B-KP (or Miwa \cite{Miwa}) and discrete C-KP (or Kashaev \cite{Kashaev-LMP,Schief-CKP}) equations, whose root lattice structure was presented in~\cite{Doliwa-Desargues-reductions}.
\end{Rem}

\subsection{Quasideterminants}

In this Section we recall the definition and basic properties of quasideterminants~\cite{Quasideterminants-GR1}. 
Results of this section are well known (in our presentation we will follow closely \cite{Quasideterminants}), however Lemma~\ref{lem:quasi-NC-H} is new.
\begin{Def}
	Given square matrix $X=(x_{ij})_{i,j=1,\dots,n}$ with formal entries $x_{ij}$. In the free division ring~\cite{Cohn} generated by the set $\{ x_{ij}\}_{i,j=1,\dots,n}$ consider the formal inverse matrix $Y=X^{-1}= (y_{ij})_{i,j=1,\dots,n}$ to $X$.
	The $(i,j)$th quasideterminant $|X|_{ij}$ of $X$ is the inverse $(y_{ji})^{-1}$ of the $(j,i)$th element of $Y$, and is often written explicitly as
	\begin{equation}
		|X|_{ij} = \left| \begin{matrix}
			x_{11} & \cdots & x_{1j} & \cdots & x_{1n} \\
			\vdots &        & \vdots &        & \vdots \\
			x_{i1} & \cdots & \boxed{x_{ij}} & \cdots & x_{in}  \\
			\vdots &        & \vdots &        & \vdots \\
			x_{n1} & \cdots & x_{nj} & \cdots & x_{nn} 
		\end{matrix} \right| .
	\end{equation}	
\end{Def}

Quasideterminants can be computed using the following recurrence relation. When $n=1$ then the quasideterminant of $X = (x)$ is just the single matrix entry $|X|_{11}=x$. For $n\geq 2$ let $X^{ij}$ be the square matrix obtained from $X$ by deleting the $i$th row and the $j$th column (with index $i/j$ skipped from the row/column enumeration), then
\begin{equation} \label{eq:QD-exp}
	|X|_{ij} = 
	x_{ij} - \sum_{\substack{ i^\prime \neq i \\ j^\prime \neq j }} x_{i j^\prime} (|X^{ij}|_{i^\prime j^\prime })^{-1} x_{i^\prime j}
\end{equation}
provided all terms in the right-hand side are defined.
\begin{Rem}
	When the elements of the matrix $X$ commute between themselves, what we denote by placing the letter $c$ over the equality sign, then the familiar matrix inversion formula gives
	\begin{equation} \label{eq:qdet-det}
		|X|_{ij} \stackrel{c}{=} (-1)^{i+j}\frac{\det X}{\det X^{ij}}.
	\end{equation}
\end{Rem}
\begin{Ex}
	Quasideterminants of generic $2\times 2$ matrix
	\begin{equation*}
		X = \left( \begin{matrix}
			x_{11} & x_{12} \\
			x_{21} & x_{22}
		\end{matrix}\right)
	\end{equation*}	
	read as follows
	\begin{gather*}
		|X|_{11} =  \left| \begin{matrix}
		x_{11} & x_{12} \\
			x_{21} & x_{22}
		\end{matrix}\right| = x_{11} - x_{12} x_{22}^{-1} x_{21}, \quad 
		|X|_{12} = \left| \begin{matrix}
		 x_{11}	 & \boxed{x_{12}} \\
			x_{21} & x_{22}
		\end{matrix}\right| = x_{12} - x_{11} x_{21}^{-1} x_{22}, \\ 
		|X|_{21} = \left| \begin{matrix}
		x_{11} & x_{12} \\
				\boxed{x_{21}} & x_{22}
		\end{matrix}\right| = x_{21} - x_{22} x_{12}^{-1} x_{11}, \quad 
		|X|_{22} = \left| \begin{matrix}
			x_{11} & x_{12} \\
			x_{21} &\boxed{x_{22}}
		\end{matrix}\right| = x_{22} - x_{21} x_{11}^{-1} x_{12}.
	\end{gather*}
\end{Ex}

Let us collect basic properties of the quasideterminants which will be used in the sequel. We start from \emph{row and column operations}:
\begin{itemize}
	\item 
	A permutation of the rows or columns of a quasideterminant does not change its value.
	\item Let the matrix $\tilde{X}$ be obtained from the matrix $X$ by multiplying the $k$th row by the element $\lambda$ of the division ring from the left, then 
	\begin{equation}
		|\tilde{X}|_{ij} = \begin{cases} \lambda|X|_{ij} & \text{if} \quad i = k, \\
			|X|_{ij} & \text{if} \quad i\neq k \quad \text{and} \; \lambda \; \text{is invertible} . \end{cases}
	\end{equation}
	
	\item Let the matrix $\hat{X}$ be obtained from the matrix $X$ by multiplying the $k$th column by the element $\mu$ of the division ring from the right, then 
	\begin{equation}
		|\hat{X}|_{ij} = \begin{cases} |X|_{ik} \, \mu & \text{if} \quad j = k, \\
			|X|_{ij} & \text{if} \quad j\neq k \quad \text{and} \; \mu \; \text{is invertible} . \end{cases}
	\end{equation}
	
	\item
	Let the matrix $\tilde{X}$ be constructed by adding to some row of the matrix $X$ its $k$th row multiplied by a scalar $\lambda$ from the left, then
	\begin{equation}
		|X|_{ij} = |\tilde{X}|_{ij}, \qquad i = 1, \dots , k-1, k+1, \dots , n, \quad j=1,\dots , n.
	\end{equation}
	
	\item Let the matrix $\hat{X}$ be constructed by addition to some column of the matrix $X$ its $l$th column multiplied by a scalar $\mu$ from the right, then
	\begin{equation}
		|X|_{ij} = |\hat{X}|_{ij}, \qquad i = 1, \dots , n, \quad j=1,\dots , l-1, l+1 , \dots ,n.
	\end{equation}
	
\end{itemize}

The second group of properties, called \emph{the homological relations}, establishes relations between quasideterminants of elements in the same row or column.
\begin{itemize}
	\item Row homological relations:
	\begin{equation} \label{eq:row-hom}
		-|X|_{ij} \cdot |X^{i k}|_{sj}^{-1} = |X|_{ik} \cdot |X^{ij}|_{sk}^{-1}, \qquad s\neq i.
	\end{equation}
	\item Column homological relations:
	\begin{equation} \label{eq:chr}
		- |X^{kj}|_{is}^{-1} \cdot |X|_{ij}=  |X^{ij}|_{ks}^{-1} \cdot |X|_{kj} , \qquad s\neq j.
	\end{equation}
\end{itemize}
\begin{Ex}
	The row homological relations between quasideterminants of elements $x_{32}$ and $x_{33} $ in the third row with auxiliary elements from the second row read
\begin{equation}
	\left| \begin{matrix}
		x_{11} &  x_{12} & x_{13} \\
		x_{21} & x_{22} & x_{23}  \\
		x_{31} & \boxed{x_{32}} & x_{33} 
	\end{matrix} \right| 
	 \left| \begin{matrix}
	x_{11} &  x_{12}  \\
	x_{21} & \boxed{x_{22}  } 
\end{matrix} \right|^{-1} = - 	\left| \begin{matrix}
x_{11} &  x_{12} & x_{13} \\
x_{21} & x_{22} & x_{23}  \\
x_{31} & x_{32} & \boxed{x_{33} }
\end{matrix} \right| 
 \left| \begin{matrix}
x_{11} &  x_{13}  \\
x_{21} & \boxed{x_{23}  } 
\end{matrix} \right|^{-1}.
\end{equation}		
	
\end{Ex}

Finally we present \emph{Sylvester's identity} for quasideterminants. 
Let $X_0 = (x_{ij})$, $i,j = 1,\dots ,k$, be a submatrix of $X$ that is invertible. For $p,q = k+1,\dots ,n$ set
\begin{equation*}
	c_{pq} = \begin{vmatrix}
		&&& x_{1q} \\
		& X_0 & & \vdots\\
		&&& x_{kq} \\
		x_{p1} & \dots & x_{pk} & \boxed{x_{pq}} 
	\end{vmatrix} \; ,
\end{equation*}
and consider the $(n-k) \times (n-k)$ matrix $C = (c_{pq})$, $p,q = k+1,\dots , n$. Then for $i,j = k+1,\dots , n$,
\begin{equation}
	|X|_{ij} = |C|_{ij} \; .
\end{equation}
We will use its simplest version for two auxiliary rows and columns, which in the commutative case reduces to the Desnanot--Jacobi identity, known also as the Dodgson condensation rule.

\begin{Ex} 	For $n=4$ and $k=2$ we have, for example,
	\begin{equation}
	\left| \begin{matrix}
		x_{11} &  x_{12} & x_{13} & x_{14}\\
		x_{21} & x_{22} & x_{23} & x_{24} \\
		x_{31} & x_{32} & x_{33} & x_{34} \\
		x_{41} & x_{42} & x_{43} & \boxed{x_{44} }	
	\end{matrix} \right| =
	\left| \begin{matrix}  	\left| \begin{matrix}
		x_{11} &  x_{12} & x_{13} \\
		x_{21} & x_{22} & x_{23}  \\
		x_{31} & x_{32} & \boxed{x_{33} }
	\end{matrix} \right|  &
\left| \begin{matrix}
	x_{11} &  x_{12} & x_{14} \\
	x_{21} & x_{22} & x_{24}  \\
	x_{31} & x_{32} & \boxed{x_{34} }
\end{matrix} \right| \\
\left| \begin{matrix}
	x_{11} &  x_{12} & x_{13} \\
	x_{21} & x_{22} & x_{23}  \\
	x_{41} & x_{42} & \boxed{x_{43} }
\end{matrix} \right|  &
\boxed{
\left| \begin{matrix}
	x_{11} &  x_{12} & x_{14} \\
	x_{21} & x_{22} & x_{24}  \\
	x_{41} & x_{42} & \boxed{x_{44} }
\end{matrix} \right| }
	\end{matrix} \right| 
\end{equation}		

\end{Ex}
\begin{Rem}
	The Sylvester identity is usually used in conjunction with row/column permutations.
\end{Rem}
Finally, we present the following consequence of the Sylvester identity, which will be used in the next section to provide large class of solutions to the non-commutative Hirota system.
\begin{Lem} \label{lem:quasi-NC-H}
	Let $A$, $B$, $C$ be row vectors of length $n$, $D$ be square matrix of size $n\times n$, and $E_1$, $E_2$ be column vectors of height $n$, then
\begin{equation}
		\left| \begin{matrix}
		A & \boxed{ a_1} \\
		D & E_1 
	\end{matrix} \right|^{-1} 	\left| \begin{matrix}
		A &  a_1 & \boxed{a_2} \\
		B & b_1 & b_2  \\
		D & E_1 &E_2 
	\end{matrix} \right| +
	\left| \begin{matrix}
	B & \boxed{ b_1} \\
	D & E_1 
\end{matrix} \right|^{-1} 	\left| \begin{matrix}
	B &  b_1 & \boxed{b_2} \\
	C & c_1 & c_2  \\
	D & E_1 &E_2 
\end{matrix} \right| +
	\left| \begin{matrix}
	C & \boxed{ c_1} \\
	D & E_1 
\end{matrix} \right|^{-1} 	\left| \begin{matrix}
	A &  a_1 & a_2 \\
	C & c_1 & \boxed{c_2}  \\
	D & E_1 &E_2 
\end{matrix} \right| =0,
\end{equation}
provided that all the quasideterminants above exist.	
\end{Lem}
\begin{proof}
Calculate the quasideterminant
\begin{equation*}
	\left| \begin{matrix}
		A &  a_1 & a_2 & \boxed{0} \\
		B & b_1 & b_2 & 0  \\
		C & c_1 & c_2 & 1  \\
		D & E_1 &E_2 & \boldsymbol{0}
	\end{matrix} \right| 
\end{equation*}
using Sylvester's identity with respect to first two rows and the two last columns; here $\boldsymbol{0}$ is the column of an appropriate height filled with zeros. On the other hand, column homological relations allow to simplify the quasideterminant, what leads to
\begin{equation*}
-	\left| \begin{matrix}
		A &  a_1 & \boxed{a_2 } \\
		B & b_1 & b_2   \\
		D & E_1 &E_2 
	\end{matrix} \right| 
	\left| \begin{matrix}
	B &  b_1 & \boxed{b_2 } \\
	C & c_1 & c_2   \\
	D & E_1 &E_2 
\end{matrix} \right|^{-1} = 	\left| \begin{matrix}
A &  a_1  & \boxed{0} \\
C & c_1 &  1  \\
D & E_1  & \boldsymbol{0}
\end{matrix} \right|  - 	\left| \begin{matrix}
A &  a_1 &\boxed{ a_2 } \\
C & c_1 & c_2  \\
D & E_1 &E_2 
\end{matrix} \right| 	\left| \begin{matrix}
B & b_1 & \boxed{b_2   }\\
C & c_1 & c_2   \\
D & E_1 &E_2 
\end{matrix} \right|^{-1} 	\left| \begin{matrix}
B & b_1 & \boxed{0 } \\
C & c_1 & 1  \\
D & E_1  & \boldsymbol{0}
\end{matrix} \right| .
\end{equation*}
The final conclusion follows by application of appropriate column homological relations once again.
\end{proof}
By transposition we obtain the following dual result. The Reader can prove it directly.
\begin{Cor} \label{cor:quasi-NC-H} 
	Let $A$, $B$, $C$ be column vectors of height $n$, $D$ be square matrix of size $n\times n$, and $E_1$, $E_2$ be row vectors of length $n$, then
\begin{equation}
		\left| \begin{matrix}
		A &  B & D  \\
		a_1 & b_1 &  E_1  \\
		\boxed{a_2} & b_2&E_2 
	\end{matrix} \right| \left| \begin{matrix}
	A &  D\\
	\boxed{ a_1} & E_1 
\end{matrix} \right|^{-1}  +
		\left| \begin{matrix}
		B & C  & D \\
		  b_1& c_1 & E_1  \\
	\boxed{b_2}	 & c_2 &E_2 
	\end{matrix} \right| \left| \begin{matrix}
	B & D \\
	\boxed{ b_1} & E_1 
\end{matrix} \right|^{-1} +
	\left| \begin{matrix}
		A &  C & D \\
		a_1 & c_1 & E_1  \\
		a_2 & \boxed{c_2} &E_2 
	\end{matrix} \right| 	\left| \begin{matrix}
	C & D \\
	\boxed{ c_1} & E_1 
\end{matrix} \right|^{-1}  =0,
\end{equation}
provided that all the quasideterminants above exist.
\end{Cor}

\section{Non-commutative multiple bi-orthogonal polynomials} \label{sec:NC-multiple-BOP}
Consider $r$ infinite matrices $(\nu^{(k)}_{ij})_{i,j \in \NN_0}$, $k=1,\dots ,r$, which provide data for the solution, referred to in the following part of the work as formal bi-moments. 
Let $\mathcal{R}[x]$ be the ring of polynomials in the commutative indeterminate $x$, with coefficients in the free field $\mathcal{R}$ generated by the bi-moments. 
Define $r$ bilinear forms $\left( \cdot \, ,  \, \cdot \right)_{(k)}$, $k=1,\dots ,r$, in $\mathcal{R}[x]$ by 
\begin{equation} \label{eq:s-forms}
	\Bigl( \sum_i a_i x^i , \sum_j b_j x^j \Bigr)_{(k)} = \sum_{i,j} a_i \nu_{i,j}^{(k)} b_j.
\end{equation}
\begin{Rem}
The results presented in this Section may be interpreted as formal theory of matrix multiple bi-orthogonal polynomials originating from appropriate matrix-valued measures $d\mu^{(k)}(x,y)$ on the plane, where the bi-moments are defined as $\nu_{i,j}^{(k)} = \int x^i y^j d\mu^{(k)}(x,y)$, where however more complicated formulas are possible as well. This makes the approach related to the so called direct linearization framework~\cite{FokasAblowitz,Nijhoff,FuNijhoff} of the theory of integrable systems or to the non-local $\bar{\partial}$-dressing method~\cite{AYF,ZakhMan,DMS}, which generalize the Riemann-Hilbert problem technique.

\end{Rem}

Given $\bs = (s_1, \dots , s_r) \in \NN_0^r $, denote 
$|\bs | = s_1 + \dots + s_r$, and define $|\bs| \times (|\bs|+1)$ matrix $\mathcal{M}_{\bs}$ consisting of $r$ row-blocks being appropriate parts $(\nu^{(k)}_{ij})_{\, 0\leq i \leq s_k-1, \, 0\leq j \leq |\bs|}$ of the bi-moment matrices
\begin{equation*}
	\mathcal{M}_{\bs} = \left( \begin{matrix}
		\nu_{0,0}^{(1)} & \nu_{0,1}^{(1)} &  \dots & \nu_{0,|\bs|}^{(1)}\\
		\vdots & \vdots & \ddots   & \vdots \\
		\nu_{s_1-1,0}^{(1)} & \nu_{s_1 -1,1}^{(1)} &  \dots & \nu_{s_1 - 1,|\bs|}^{(1)} \\
		\cdots & \cdots &   & \cdots \\
		\nu_{0,0}^{(r)} & \nu_{0,1}^{(r)} &  \dots & \nu_{0,|\bs|}^{(r)}\\
		\vdots & \vdots & \ddots   & \vdots \\
		\nu_{s_r-1,0}^{(r)} & \nu_{s_r -1,1}^{(r)} &  \dots & \nu_{s_r - 1,|\bs|}^{(r)} 
	\end{matrix} \right) \; .
\end{equation*} 
Supplement the matrix $\mathcal{M}_{\bs}$ by the row
$ (1, x, x^2,  \dots , x^{|\bs|}) $, and define monic polynomial $Q_{\bs}(x)$ as quasideterminant 
of such square matrix with respect to the rightmost element of that additional row
\begin{equation} \label{eq:pol-Q}
	Q_{\bs}(x) = \left| \begin{matrix}
		\nu_{0,0}^{(1)} & \nu_{0,1}^{(1)} &  \dots & \nu_{0,|\bs|}^{(1)}\\
		\vdots & \vdots & \ddots   & \vdots \\
		\nu_{s_1-1,0}^{(1)} & \nu_{s_1 -1,1}^{(1)} &  \dots & \nu_{s_1 - 1,|\bs|}^{(1)} \\
		\cdots & \cdots &   & \cdots \\
		\nu_{0,0}^{(r)} & \nu_{0,1}^{(r)} &  \dots & \nu_{0,|\bs|}^{(r)}\\
		\vdots & \vdots & \ddots   & \vdots \\
		\nu_{s_r-1,0}^{(r)} & \nu_{s_r -1,1}^{(r)} &  \dots & \nu_{s_r - 1,|\bs|}^{(r)}  \\
		1 & x &  \dots & \boxed{x^{|\bs|}}
	\end{matrix} \right| \, .
\end{equation}

\begin{Prop}
	The polynomials $Q_{\bs}(x)$ satisfy the following multiple orthogonality relations
	\begin{equation} \label{eq:mult-orthog-Q}
		\bigl( x^i , Q_{\bs}(x) \bigr)_{(k)} = 0, \qquad i = 0, \dots , s_k -1.
	\end{equation}
\end{Prop}
\begin{proof}
	By \eqref{eq:s-forms}, the left hand side of equation \eqref{eq:mult-orthog-Q} equals to the quasideterminant (with respect to the rightmost element o the last row) of the matrix $\mathcal{M}_{\bs}$ supplemented by the row
	$ (	\nu_{i,0}^{(k)} , \nu_{i,1}^{(k)} ,  \dots , \nu_{i,|\bs|}^{(k)}) $. In the considered range of $i$ the quasideterminant has two identical rows, and therefore vanishes.
\end{proof}

\begin{Cor}
Supplement the matrix $\mathcal{M}_{\bs}$ by the row $ (\nu_{s_k,0}^{(k)} , \nu_{s_k,1}^{(k)} ,  \dots , \nu_{s_k,|\bs|}^{(k)}) $, and define the corresponding quasideterminant
\begin{equation} \label{eq:rho}
	\rho^{(k)}_{\bs} = \left| \begin{matrix}
		\nu_{0,0}^{(1)} & \nu_{0,1}^{(1)} &  \dots & \nu_{0,|\bs|}^{(1)}\\
		\vdots & \vdots & \ddots   & \vdots \\		\nu_{s_1-1,0}^{(1)} & \nu_{s_1 -1,1}^{(1)} &  \dots & \nu_{s_1 - 1,|\bs|}^{(1)} \\
		\cdots & \cdots &   & \cdots \\
		\nu_{0,0}^{(r)} & \nu_{0,1}^{(r)} &  \dots & \nu_{0,|\bs|}^{(r)}\\
		\vdots & \vdots & \ddots   & \vdots \\		
		\nu_{s_r-1,0}^{(r)} & \nu_{s_r -1,1}^{(r)} &  \dots & \nu_{s_r - 1,|\bs|}^{(r)} \\
		\nu_{s_k,0}^{(k)} & \nu_{s_k ,1}^{(k)} &  \dots & \boxed{\nu_{s_k,|\bs|}^{(k)} }
	\end{matrix} \right| \; , \quad k=1,\dots , r,
\end{equation}
then 
\begin{equation} \label{eq:mult-norm-Q}
	\bigl( x^{s_k} , Q_{\bs}(x) \bigr)_{(k)} = \rho^{(k)}_{\bs} , \qquad k=1,\dots , r .
\end{equation}
\end{Cor}
\begin{Rem}
	The matrix used to define $\rho^{(k)}_{\bs}$ can be described also (up to a permutation of rows, which does not change the quasideterminants) as $\mathcal{M}_{\bs+\be_k}$ with the last column removed.
\end{Rem}

\begin{Prop} \label{prop:NC-rho-Q}
	When $r\geq 3$ then the quasideterminants $\rho^{(k)}_{\bs}$ satify the potential form  \eqref{eq:NC-H-rho} of the non-commutative Hirota system.
Moreover, the polynomials $Q_{\bs}(x)$ satisfy the corresponding linear system
\begin{equation} \label{eq:lin-NC-H-rho}
	Q_{\bs + \be_i} (x) - Q_{\bs + \be_j} (x) = Q_{\bs} (x) [\rho^{(i)}_{\bs}]^{-1}	\rho^{(i)}_{\bs+\be_j}, \qquad i\neq j \, .
\end{equation}	
\end{Prop}
\begin{proof}
The first part of equations \eqref{eq:NC-H-rho} are just the appropriate column homological relations applied to the matrix 
$\mathcal{M}_{\bs + \be_i + \be_j}$ with the last column removed. The second part of \eqref{eq:NC-H-rho} follows directly from Lemma~\ref{lem:quasi-NC-H} with appropriate identification and permutation of rows. The same argument together with column homological relations leads to \eqref{eq:lin-NC-H-rho}.
\end{proof}

Analogous results can be given for the $(|\bs| + 1) \times |\bs|$ column-block matrix 
\begin{equation*}
	\mathcal{N}_{\bs} = \left( \begin{matrix}
		\nu_{0,0}^{(1)} & \cdots & \nu_{0,s_{1}-1 }^{(1)} & \vdots & \nu_{0,0}^{(r)} &\dots & \nu_{0,s_r-1}^{(r)}\\
		\nu_{1,0}^{(1)} & \cdots & \nu_{1,s_{1}-1 }^{(1)} & \vdots & \nu_{1,0}^{(r)} &\dots & \nu_{1,s_r-1}^{(r)} \\
		\vdots & \ddots & \vdots &  & \vdots & \ddots & \vdots  \\
		\nu_{|\bs|,0}^{(1)} & \cdots & \nu_{|\bs|,s_{1}-1 }^{(1)} & \vdots & \nu_{|\bs|,0}^{(r)} &\dots & \nu_{|\bs|,s_r-1}^{(r)} 		
	\end{matrix} \right) \; ,
\end{equation*} 
the polynomials 
\begin{equation} \label{eq:pol-P}
	P_{\bs}(x) = \left| \begin{matrix}
		\nu_{0,0}^{(1)} & \cdots & \nu_{0,s_{1}-1 }^{(1)} & \vdots & \nu_{0,0}^{(r)} &\dots & \nu_{0,s_r-1}^{(r)} & 1 \\
	\nu_{1,0}^{(1)} & \cdots & \nu_{1,s_{1}-1 }^{(1)} & \vdots & \nu_{1,0}^{(r)} &\dots & \nu_{1,s_r-1}^{(r)} & x& \\
	\vdots & \ddots & \vdots &  & \vdots & \ddots & \vdots & \vdots \\
	\nu_{|\bs|,0}^{(1)} & \cdots & \nu_{|\bs|,s_{1}-1 }^{(1)} & \vdots & \nu_{|\bs|,0}^{(r)} &\dots & \nu_{|\bs|,s_r-1}^{(r)} & \boxed{x^{|\bs|}}
	\end{matrix} \right| \, ,
\end{equation}
and the corresponding potentials 
\begin{equation} \label{eq:pi}
	\pi_{\bs}^{(k)} = \left| \begin{matrix}
		\nu_{0,0}^{(1)} & \cdots & \nu_{0,s_{1}-1 }^{(1)} & \vdots & \nu_{0,0}^{(r)} &\dots & \nu_{0,s_r-1}^{(r)} & \nu_{0,s_k}^{(k)} \\
		\nu_{1,0}^{(1)} & \cdots & \nu_{1,s_{1}-1 }^{(1)} & \vdots & \nu_{1,0}^{(r)} &\dots & \nu_{1,s_r-1}^{(r)} & \nu_{1,s_k}^{(k)} & \\
		\vdots & \ddots & \vdots &  & \vdots & \ddots & \vdots & \vdots \\
		\nu_{|\bs|,0}^{(1)} & \cdots & \nu_{|\bs|,s_{1}-1 }^{(1)} & \vdots & \nu_{|\bs|,0}^{(r)} &\dots & \nu_{|\bs|,s_r-1}^{(r)} & \boxed{\nu_{|\bs|,s_k}^{(k)}}
	\end{matrix} \right| \, .
\end{equation}
By replacing in the above proofs the row/column homological relations by the column/row relations, and Lemma~\ref{lem:quasi-NC-H} by Corollary~\ref{cor:quasi-NC-H} we obtain as follows:
\begin{itemize}
	\item the multiple orthogonality relations
	\begin{equation} \label{eq:mult-orthog-Q}
	\bigl( P_{\bs}(x) , x^i  \bigr)_{(k)} = 0, \qquad i = 0, \dots , s_k -1,
\end{equation}	
	\item the normalizations
	\begin{equation} \label{eq:mult-norm-P}
	\bigl( P_{\bs}(x) , x^{s_k}  \bigr)_{(k)} = \pi^{(k)}_{\bs} , \qquad k=1,\dots , r ,
	\end{equation}
\item the potential form of the transposed non-commutative Hirota system
\begin{equation} 
	\label{eq:NC-H-pi}
	\pi^{(i)}_{\bs+\be_j} 	[\pi^{(i)}_{\bs}]^{-1} + 	\pi^{(j)}_{\bs+\be_i} [\pi^{(j)}_{\bs}]^{-1} = 0, \qquad 
		\pi^{(i)}_{\bs+\be_j} [\pi^{(i)}_{\bs}]^{-1} + 
		\pi^{(j)}_{\bs+\be_k} [\pi^{(j)}_{\bs}]^{-1} + 
		\pi^{(k)}_{\bs+\be_i} 	[\pi^{(k)}_{\bs}]^{-1} = 0,
\end{equation}
with distinct $i,j,k$, 
\item the corresponding transposed linear system
\begin{equation} \label{eq:lin-NC-H-pi}
	P_{\bs + \be_i} (x) - P_{\bs + \be_j} (x) = \pi^{(i)}_{\bs+\be_j} [\pi^{(i)}_{\bs}]^{-1} P_{\bs} (x) 	, \qquad i\neq j \, .
\end{equation}	
\end{itemize}

\begin{Rem}
	The following two natural specifications/reductions do not appear to have been studied yet:
	\begin{itemize}
		\item When we allow the bi-moments to commute between themselves one obtains multiple bi-orthogonal polynomials together with the corresponding solutions of the Hirota equations~\eqref{eq:Hirota-com}. 
		\item In the case $r=1$ we end up with non-commutative bi-orthogonal polynomials.
	\end{itemize}
\end{Rem}

\section{Non-commutative multiple orthogonal polynomials}
\label{sec:NC-multiple-OP}
In the present Section we transfer to the non-commutative setting the structural results of the theory of multiple orthogonal polynomials and their relation to integrable systems theory. In doing that we are guided by the recent work \cite{Doliwa-NHP} on the non-commutative Hermite--Pad\'{e} approximation theory.

\subsection{Definition of the non-commutative multiple orthogonal polynomials}
Following the algebraic theory of non-commutative orthogonal polynomials given by \cite{NCSF} let us define the multiple non-commutative  orthogonal polynomials.
Given $r$ sequences $(\nu^{(k)}_j)_{j\in \NN_0}$, $k=1,\dots , r$, of moments, consider the free division ring $\mathcal{R}$ generated by the moments together with a natural involutive \emph{anti-automorphism} $(\;)^*$ such that $(\nu^{(k)}_j)^* = \nu^{(k)}_j$. The anti-automorphism can be extended to the ring of polynomials  $\mathcal{R}[x]$ in the commutative indeterminate $x$ with coefficients in  $\mathcal{R}$ by putting $\bigl( \sum_i a_i x^i \bigr)^* = \sum_i a_i^* x^i $. Define $r$ sesquilinear forms $\left< . \, , . \right>e_{(k)}$, $k=1,\dots , r$, by the relations
\begin{equation}
	\Bigl \langle \sum_i a_i x^i , \sum_j b_j x^j \Bigr \rangle_{(k)} = \sum_{i,j} a_i^* \nu_{i+j}^{(k)} b_j .
\end{equation}

Given $r$ non-negative integers $\bs = (s_1,\dots , s_r) \in\mathbb{N}_0^r$. The \emph{multiple orthogonal polynomial of the index $\bs$} is the monic polynomial $Q_{\bs}(x)$  of degree $|\bs|$ 
which satisfies the following orthogonality conditions 
\begin{equation} \label{eq:NC-orth-r}
	\bigl \langle x^j , Q_{\bs} (x) \bigr \rangle_{(k)} = 0, \qquad j=0,1,\dots s_k -1,
\end{equation}
whenever such a polynomial exists and is unique. The equations \eqref{eq:NC-orth-r} give a system
of $|\bs|$ linear equations for the $|\bs|$ non-leading coefficients of the polynomial $Q_{\bs}(x)$;
the multi-index $\bs$ is said to be \emph{normal} whenever the solution exists and is unique. A system of moments is said to be perfect in case all $\bs \in \NN^r$ are normal; we restrict attention to perfect systems only.
One can observe that above-mentioned linear equations can be solved in terms of quasideterminants, what gives
	\begin{equation} \label{eq:Q-r}
		Q_{\bs}(x) = \left| \begin{matrix}
			\nu_{0}^{(1)} & \nu_{1}^{(1)}  &  \dots & \nu_{|\bs|}^{(1)} \\
			\vdots & \vdots & \ddots   & \vdots \\
			\nu_{s_1 - 1}^{(1)}  & \nu_{s_1}^{(1)}  &  \dots & \nu_{ |\bs|+ s_1 - 1}^{(1)}  \\
			\cdots & \cdots &   & \cdots \\
			\nu_{0}^{(r)} & \nu_{1}^{(r)} &  \dots & \nu_{|\bs|}^{(r)}\\
			\vdots & \vdots & \ddots   & \vdots \\
			\nu_{s_r - 1}^{(r)} & \nu_{s_r}^{(r)} &  \dots & \nu_{ |\bs| + s_r - 1}^{(r)} \\
			1 & x & \dots & \boxed{x^{|\bs|}}
		\end{matrix} 	\right|.
	\end{equation}	
Moreover we have
	\begin{equation} \label{eq:NC-orth-rho}
		\bigl \langle x^{s_k} , Q_{\bs} (x) \bigr \rangle_{(k)} = \rho^{(k)}_{\bs},
	\end{equation}
	where 
	\begin{equation} \label{eq:rho-pol}
		\rho^{(k)}_{\bs} = \left| \begin{matrix}
			\nu_{0}^{(1)} & \nu_{1}^{(1)}  &  \dots & \nu_{|\bs|}^{(1)} \\
			\vdots & \vdots & \ddots   & \vdots \\
			\nu_{s_1 - 1}^{(1)}  & \nu_{s_1}^{(1)}  &  \dots & \nu_{ |\bs|+ s_1 - 1}^{(1)}   \\
			\cdots & \cdots &   & \cdots \\
			\nu_{0}^{(r)} & \nu_{1}^{(r)} &  \dots & \nu_{|\bs|}^{(r)}\\
			\vdots & \vdots & \ddots   & \vdots \\
			\nu_{s_r - 1}^{(r)} & \nu_{s_r}^{(r)} &  \dots & \nu_{ |\bs| + s_r - 1}^{(r)} \\
			\nu_{s_k}^{(k)}  & \nu_{s_k+1}^{(k)}  &  \dots & \boxed{\nu_{ |\bs|+ s_k}^{(k)} }
		\end{matrix} 	\right|, \quad k=1,\dots,r .
	\end{equation}	

The above results are direct specifications of the theory of (non-commutative multiple) bi-orthogonal polynomials, given in Section~\ref{sec:NC-multiple-BOP}, to the simpler case when 
\begin{equation} \label{eq:NC-red}
	\nu^{(k)}_{i,j} = \nu^{(k)}_{i+j}	, \quad r=1,\dots ,r, \quad i,j \geq 0,		
\end{equation}
and, correspondingly,
\begin{equation}
	P_{\bs}(x) = Q_{\bs}(x)^*, \qquad \pi^{(k)}_{\bs} = \rho_{\bs}^{(k)*}
\end{equation}
In particular, by Proposition~\ref{prop:NC-rho-Q}, when $r\geq 3$ then the quasideterminants $\rho^{(k)}_{\bs}$ satisfy the potential form  \eqref{eq:NC-H-rho} of the non-commutative Hirota system, and the polynomials $Q_{\bs}(x)$ satisfy the corresponding linear system \eqref{eq:lin-NC-H-rho}. In the next Section we will present additional relations specific to the reduction~\eqref{eq:NC-red}.

\subsection{Further properties of the non-commutative multiple orthogonal polynomials}
Below we present the non-commutative generalization of the standard three-term relation for orthogonal polynomials. We also provide the constraints which relate coefficients of the generalized relations. The commutative version of the relations and of the constraints for multiple orthogonal polynomials was given in \cite{VanAsche-nn-mop}.

	\begin{Prop}
	The non-commutative multiple orthogonal polynomials satisfy the following generalization of the three term relation
	\begin{equation} \label{eq:gen-3-term}
		x Q_{\bs} (x) = Q_{\bs + \be_j}(x) +  Q_{\bs}(x) b_{\bs}^{(j)} + \sum_{k=1}^r  Q_{\bs - \be_k}(x) a_{\bs}^{(k)}, \qquad j=1,2,\dots , r,
	\end{equation}
	where
	\begin{equation} \label{eq:NC-orth-a}
		a_{\bs}^{(k)} = [\rho^{(k)}_{\bs - \be_k}]^{-1} \rho^{(k)}_{\bs}.
	\end{equation}
\end{Prop}
\begin{proof}
By normality, $x Q_{\bs}(x) $ and $Q_{\bs+\be_j}(x)$ are monic polynomials of degree $|\bs|+1$, so their difference is of degree not greater than $|\bs|$.  Hence there is a uniquely determined
constant $b^{(j)}_{\bs}$, which will be specified later, such that the polynomial
\begin{equation*}
T_{\bs}^{(j)}(x) = x Q_{\bs}(x) -Q_{\bs+\be_j}(x) - Q_{\bs}(x) b^{(j)}_{\bs},
\end{equation*}  is of degree not greater than $|\bs|-1$. Note that $T_{\bs}^{(j)}(x) $ is a solution to the system of equations
\begin{equation*} 
\bigl 	\langle x^i , T_{\bs}^{(j)} (x) \bigr \rangle_{(k)} = 0, \qquad i=0,1,\dots s_k -2, \quad k=1,\dots ,r.
\end{equation*}
By normality, the space of solutions is of dimension $r$ with the basis $\{Q_{\bs - \be_i}(x)\}_{i=1,\dots , r}$, thus
\begin{equation} \label{eq:xQ-b-a}
x Q_{\bs}(x) -Q_{\bs+\be_j}(x) - Q_{\bs}(x) b^{(j)}_{\bs}  = \sum_{i=1}^r Q_{\bs - \be_i}(x) a^{(i,j)}_{\bs}  \, .
\end{equation}
By taking $k$-th scalar product of both sides of \eqref{eq:xQ-b-a} with $x^{s_k-1}$  we get 
\begin{equation*}
	\bigl \langle x^{s_k}, Q_{\bs}(x) \bigr \rangle_{(k)} = \bigl \langle x^{s_k-1}, Q_{\bs-\be_k}(x) \bigr \rangle_{(k)} \; a^{(k,j)}_{\bs},
\end{equation*}
what implies that the coefficients $a^{(k,j)}_{\bs}$ are independent of the second upper index $j$, what due to equations \eqref{eq:NC-orth-rho} gives relations \eqref{eq:NC-orth-a}. 
\end{proof}

\begin{Cor}
	By considering in equation \eqref{eq:gen-3-term} coefficients at $x^{|\bs|}$ we get
\begin{equation}
	\qquad b_{\bs}^{(j)} = [\rho^{(j)}_{\bs}]^{-1} \tilde{\rho}^{(j)}_{\bs} -  [\rho^{(j)}_{\bs-\be_j}]^{-1} \tilde{\rho}^{(j)}_{\bs-\be_j} ,
\end{equation}
where $\tilde{\rho}^{(k)}_{\bs}$ is the quasideterminant of the reduced, by relation \eqref{eq:NC-red},  matrix $\mathcal{M}_{\bs}$ with the $k$-th block supplemented by the row $(	\nu_{s_k+1}^{(k)} , \nu_{s_k+2}^{(k)} , \dots , \nu_{s_k + |\bs|+1}^{(k)})$, with respect to its rightmost element, i.e. 
\begin{equation} \label{eq:rho-pol}
	\tilde{\rho}^{(k)}_{\bs} = \left| \begin{matrix}
		\nu_{0}^{(1)} & \nu_{1}^{(1)}  &  \dots & \nu_{|\bs|}^{(1)} \\
		\vdots & \vdots & \ddots   & \vdots \\
		\nu_{s_1 - 1}^{(1)}  & \nu_{s_1}^{(1)}  &  \dots & \nu_{ |\bs|+ s_1 - 1}^{(1)}   \\
		\cdots & \cdots &   & \cdots \\
		\nu_{0}^{(r)} & \nu_{1}^{(r)} &  \dots & \nu_{|\bs|}^{(r)}\\
		\vdots & \vdots & \ddots   & \vdots \\
		\nu_{s_r - 1}^{(r)} & \nu_{s_r}^{(r)} &  \dots & \nu_{ |\bs| + s_r - 1}^{(r)} \\
		\nu_{s_k+1}^{(k)}  & \nu_{s_k+2}^{(k)}  &  \dots & \boxed{\nu_{ |\bs|+ s_k+1}^{(k)} }
	\end{matrix} 	\right|, \quad k=1,\dots,r .
\end{equation}	
\end{Cor}
The coefficients of the relations \eqref{eq:gen-3-term}, which are constructed using the moments, cannot be given arbitrarily, but should satisfy the compatibility conditions of  \eqref{eq:gen-3-term}. 
\begin{Prop}
	The coefficients $a_{\bs}^{(k)}$ and $b_{\bs}^{(j)}$, $j=1,\dots ,r$, satisfy the following equations	
		\begin{align} \label{eq:NC-ab-b}
		\quad \; b_{\bs+\be_k}^{(j)} - b_{\bs + \be_j}^{(k)} & = b_{\bs}^{(j)} - b_{\bs}^{(k)},\\  \label{eq:NC-ab-bb} 
		b_{\bs}^{(k)} b_{\bs + \be_k}^{(j)} - b_{\bs}^{(j)} b_{\bs + \be_j}^{(k)}  & = \sum_{i=1}^r \left( a_{\bs+\be_k}^{(i)} - a_{\bs+\be_j}^{(i)}\right),\\  \label{eq:NC-ab-ba}
		( b_{\bs -\be_j}^{(j)} - b_{\bs - \be_j}^{(k)}) a_{\bs + \be_k}^{(j)} & = a_{\bs}^{(j)}(b_{\bs}^{(j)} - b_{\bs}^{(k)}).
	\end{align}
\end{Prop}	
\begin{proof}
Notice first that equations \eqref{eq:gen-3-term} imply another form of the linear problem \eqref{eq:lin-NC-H-rho} to the Hirota system
\begin{equation} \label{eq:NC-lin-b}
Q_{\bs+\be_j}(x) -  Q_{\bs+\be_k}(x) = Q_{\bs}(x) (b^{(j)}_{\bs}- b^{(k)}_{\bs}) .
\end{equation}
By subtracting \eqref{eq:gen-3-term}  shifted in $\be_k$ direction from its $k$-version shifted in $\be_j$ direction we obtain
\begin{gather*}
	x(Q_{\bs + \be_j}(x) - Q_{\bs + \be_k}(x)) = \\ Q_{\bs + \be_j}(x)  b_{\bs+\be_j}^{(k)} - Q_{\bs + \be_k}(x)  b_{\bs+\be_k}^{(j)}  +  \sum_{i=1}^r  Q_{\bs + \be_j - \be_i}(x) a_{\bs + \be_j}^{(i)} - \sum_{i=1}^r  Q_{\bs + \be_k - \be_i}(x) a_{\bs + \be_k}^{(i)} ,
\end{gather*}
which we reduce further replacing terms of the form $Q_{\bs + \be_k}(x)$ by using equation \eqref{eq:gen-3-term}, and terms of the form $Q_{\bs + \be_k - \be_i}(x)$ by using equation \eqref{eq:NC-lin-b} shifted appropriately. Finally, equations \eqref{eq:NC-ab-b} are obtained by collecting coefficients in front of $xQ_{\bs}(x)$, equations \eqref{eq:NC-ab-bb} are obtained by collecting coefficients in front of $Q_{\bs}(x)$, and equations \eqref{eq:NC-ab-ba} are obtained by collecting coefficients in front of $Q_{\bs-\be_i}(x)$. 
\end{proof}

\section{Discrete-time evolution of the moments, and its consequences}
\label{sec:MTS}
In this Section we go out of the initial space of discrete variables by adding suitable evolution variable, what in the simplest case $r=1$ and for commuting variables gives the discrete-time Toda equations~\cite{Hirota-2dT}, as we recalled in Section~\ref{sec:QD,NCH}. For multiple orthogonal polynomials in the commutative case such evolution was studied in   \cite{AptekarevDerevyaginMikiVanAssche}.

Let us assume that the moments evolve in discrete time variable $t\in \mathbb{N}_0$ according simple equation, 
\begin{equation}
	\nu^{(k)}_{j;t} = \nu^{(k)}_{j+t}, \qquad k=1,\dots ,r,
\end{equation}
the direct generalization of equation~\eqref{eq:mu-t}.
For fixed $t$, the corresponding reduced matrix $\mathcal{M}_{\bs;t}$ given by 
\begin{equation} \label{eq:M-t-r}
	\mathcal{M}_{\bs;t} = \left(  \begin{matrix}
		\nu_{t}^{(1)} & \nu_{t+1}^{(1)}  &  \dots & \nu_{t+|\bs|}^{(1)} \\
		\vdots & \vdots & \ddots   & \vdots \\
		\nu_{t+ s_1 - 1}^{(1)}  & \nu_{t+s_1}^{(1)}  &  \dots & \nu_{ t+ |\bs|+ s_1 - 1}^{(1)}  \\
		\cdots & \cdots &   & \cdots \\
		\nu_{t}^{(r)} & \nu_{t+1}^{(r)} &  \dots & \nu_{t+|\bs|}^{(r)}\\
		\vdots & \vdots & \ddots   & \vdots \\
		\nu_{t+ s_r - 1}^{(r)} & \nu_{t+s_r}^{(r)} &  \dots & \nu_{ t+|\bs| + s_r - 1}^{(r)}
	\end{matrix} 	\right),
\end{equation}	
allows to calculate the corresponding non-commutative multiple orthogonal polynomials $Q_{\bs;,t}(x)$, similarly we have the quasideterminants $\rho^{(k)}_{\bs;t}$ and $\tilde{\rho}^{(k)}_{\bs;t}$ and the analogs $a^{(j)}_{\bs;t}$, $b^{(j)}_{\bs;t}$ of the coefficients in the corresponding versions of equations \eqref{eq:gen-3-term} and \eqref{eq:NC-ab-b}-\eqref{eq:NC-ab-ba}.
All functions and identities considered in previous Section have their fixed time analogs. We will be however interested in relations between the functions in neighboring moments of the discrete time $t$. These will provide more insight into the non-commutative multiple orthogonal polynomials and their relation to integrable systems theory.

\begin{Prop}
	When $r\geq 1$ then the multiple polynomials $Q_{\bs;t}(x)$ satisfy the following non-commutative multiple version of the linear problem \eqref{eq:lin-A}	
	\begin{equation} \label{eq:NC-Q-A-j-t}
	xQ_{\bs;t+1}(x)  = Q_{\bs + \be_j;t}(x) + Q_{\bs;t}(x) A^{(j)}_{\bs;t} , \qquad j=1,\dots r,	
	\end{equation}
	where 
	\begin{equation} \label{eq:NC-A-Q-r}
	A^{(j)}_{\bs;t} = [\rho^{(j)}_{\bs;t}]^{-1} \rho^{(j)}_{\bs;t+1}.
	\end{equation}
\end{Prop}
\begin{proof}
	Calculate the polynomial $Q_{\bs+\be_j;t}(x)$ directly as the quasideterminant of the matrix $\mathcal{M}_{\bs+\be_j;t}$ supplemented by the row $(1,x, \dots , x^{|\bs|}, x^{|\bs|+1})$, with respect to the rightmost element of that row. On the other hand, we can do the same by using the Sylvester identity with respect to the first and the last columns, the last row of the enlarged $j$th block and that last supplementing row, and after using row homological relations we get the statement. 
\end{proof}
\begin{Rem}
As the alternative proof we can observe that the difference $xQ_{\bs;t+1}(x)  - Q_{\bs + \be_j;t}(x)$ is a polynomial of degree not greater than $|\bs|$ which satisfies the same orthogonality conditions as $Q_{\bs;t}(x)$. To find  proportionality coefficient we take $j$-th scalar product in time $t$ of both sides of \eqref{eq:NC-Q-A-j-t} with $x^{s_j}$ noticing that
\begin{equation*}
	\bigl 	\langle x^{s_j}, Q_{\bs;t} \bigr \rangle_{(j),t} = \rho^{(j)}_{\bs;t}, \quad
	\bigl 	\langle x^{s_j}, Q_{\bs+\be_j;t} \bigr\rangle_{(j);t} = 0, \quad
	\bigl 	\langle x^{s_j}, x Q_{\bs;t+1} \bigr\rangle_{(j);t}  =
	\bigl 	\langle x^{s_j}, Q_{\bs;t+1} \bigr\rangle_{(j);t+1} = \rho^{(j)}_{\bs;t+1}. \qquad
\end{equation*}
\end{Rem}
\begin{Cor}
Compatibility of the equations \eqref{eq:NC-Q-A-j-t} between themselves implies  the following system
\begin{equation} \label{eq:AA-r}
	A_{\bs;t+1}^{(j)} - A_{\bs;t+1}^{(k)} = A_{\bs+\be_k;t}^{(j)} - A_{\bs+\be_j;t}^{(k)}, 
	\quad A_{\bs;t}^{(k)}  A_{\bs+\be_k;t}^{(j)} = A_{\bs;t}^{(j)} A_{\bs+\be_j;t}^{(k)} .
\end{equation}	
\end{Cor}
\begin{Cor}
	Compatibility of the constant time $t$ versions of the linear equations~\eqref{eq:lin-NC-H-rho} satisfied by the orthogonal polynomials, their shifted time $t+1$ versions	and equations \eqref{eq:NC-Q-A-j-t} supplements, using also relation~\eqref{eq:NC-A-Q-r}, the Hirota system~\eqref{eq:NC-H-rho} by analogous equations
	\begin{equation}
		\label{eq:NC-H-t-rho}
[\rho^{(i)}_{\bs;t}]^{-1} \rho^{(i)}_{\bs;t+1} - [\rho^{(j)}_{\bs;t}]^{-1} \rho^{(j)}_{\bs;t+1} + [\rho^{(i)}_{\bs;t}]^{-1} \rho^{(i)}_{\bs +\be_j;t} = 0 , \qquad i\neq j .		
	\end{equation} 
\end{Cor}
\begin{Rem}
Equation \eqref{eq:NC-H-t-rho} can be also obtained by taking take $i$-th scalar product in time $t$ of both sides of \eqref{eq:NC-Q-A-j-t} with $x^{s_i}$, $i\neq j$.
\end{Rem}

To find non-commutative multiple version of the discrete-time Toda equations \eqref{eq:dt-T-AB} let us formulate fist the corresponding version of the linear equations \eqref{eq:lin-B}.
\begin{Prop}
The multiple orthogonal polynomials evolve according to the equations
	\begin{equation} \label{eq:Q-B-t}
	Q_{\bs;t+1}(x)  = Q_{\bs;t}(x) - \sum_{i=1}^r Q_{\bs - \be_i;t+1}(x) B^{(i)}_{\bs;t} ,
	\end{equation}
	where 
\begin{equation} \label{eq:B-Q-r}
B^{(j)}_{\bs;t} = [\rho^{(j)}_{\bs - \be_j;t+1}]^{-1} \rho^{(j)}_{\bs;t}.
\end{equation}
\end{Prop}
\begin{proof}
Because
\begin{equation*}
\langle x^{i}, Q_{\bs;t} \rangle_{(j);t+1} = \langle x^{i+1}, Q_{\bs;t} \rangle_{(j);t}	= 0, \qquad i=0, \dots , s_j-2,	
\end{equation*}	
therefore the difference $Q_{\bs;t+1}(x) - Q_{\bs;t}(x)	$ is a polynomial of degree not greater than $|\bs|-1$ satisfying the orthogonality conditions as above, thus must be decomposable in the basis $\{  Q_{\bs - \be_j;t+1}(x)\}_{j=1,\dots , r}$.  To find  proportionality coefficient we take $j$-th scalar product in time $t+1$ of both sides of \eqref{eq:Q-B-t} with $x^{s_j-1}$ noticing that
\begin{gather*}
\bigl	\langle x^{s_j-1}, Q_{\bs;t+1} \bigr \rangle_{(j);t+1} =0, \quad
\bigl	\langle x^{s_j-1}, Q_{\bs;t} \bigr \rangle_{(j);t+1} = \langle x^{s_j}, Q_{\bs;t} \bigr \rangle_{(j);t} =  \rho^{(j)}_{\bs;t}, \\
\bigl	\langle x^{s_j-1}, Q_{\bs-\be_j;t+1} \bigr \rangle_{(j);t+1}  = \rho^{(j)}_{\bs-\be_j;t+1}, \qquad 
\bigl	\langle x^{s_j-1}, Q_{\bs-\be_i;t+1} \bigr \rangle_{(j);t+1}  = 0 \quad \text{for} \quad i\neq j	. 
\end{gather*}
\end{proof}
The following result can be obtained by standard calculation.
\begin{Prop}
In addition to equations \eqref{eq:AA-r}, compatibility of the equations \eqref{eq:NC-Q-A-j-t} and \eqref{eq:Q-B-t} implies the following system	
\begin{align}
	\label{eq:AB-BA-r}
	A_{\bs-\be_j;t+1}^{(j)} B_{\bs;t+1}^{(j)} & =  B_{\bs;t}^{(j)}  A_{\bs;t}^{(j)}, \\
	\label{eq:BAA-r}
	B^{(k)}_{\bs;t}\left( A^{(k)}_{\bs;t} - A^{(j)}_{\bs;t} \right) &  = \left( A^{(k)}_{\bs-\be_k;t+1} - A^{(j)}_{\bs-\be_k;t+1} \right) B^{(k)}_{\bs + \be_j;t},\\
	\label{eq:A+B-r}
	A_{\bs;t+1}^{(j)} + \sum_{i=1}^r B_{\bs;t+1}^{(i)} & = A_{\bs;t}^{(j)} + \sum_{i=1}^r B_{\bs+\be_j;t}^{(i)}. 
\end{align}	
\end{Prop}

\begin{Rem}
	The commutative version of the linear system \eqref{eq:NC-Q-A-j-t}, \eqref{eq:Q-B-t} satisfied by the multiple orthogonal polynomials and   the corresponding compatibility conditions \eqref{eq:AA-r}, \eqref{eq:AB-BA-r}--\eqref{eq:A+B-r} was given in	 \cite{AptekarevDerevyaginMikiVanAssche}, see also \cite{Doliwa-MOP} for its determinantal interpretation.
\end{Rem}

\section{Final remarks}

We constructed the general theory of non-commutative multiple bi-orthogonal polynomials based on quasideterminantal identities, and we have presented it within the context of integrable systems theory. In particular, we have obtained new class of solutions to the non-commutative Hirota system. The analogs of the spectral data are provided by the (formal) bi-moments. This allows to consider such generalized theory of orthogonal polynomials as a part of the integrable systems theory.

The fact that orthogonal polynomials and integrable systems are deeply connected is not a novelty. Our work provides its explanation on the level of the most general non-commutative discrete integrable system known in the literature. 
On the other hand, the approach of this paper with matrices of formal bi-moments considered as primary objects, can be extended to bi-infinite matrices. This gives new insight into the theory of integrable systems generalizing the so called direct method by Hirota. It offers, for example, explanation of the  "nonlinear" superposition formula in terms of the "linear" sum of bi-moment matrices. Moreover, various symmetry properties imposed on the bi-moment matrices should produce the corresponding reductions of the non-commutative Hirota system, as  the non-commutative multidimensional Toda lattice equations given in the present paper.  


\end{document}